\documentclass{llncs}

\usepackage{amssymb,amsmath,graphicx}
\usepackage{pgf}
\usepackage{color,enumerate,comment,boxedminipage}
\usepackage{float}
\usepackage{hyperref}
\usepackage{comment}
\usepackage{tikz}
\usepackage{subcaption} 

\newcommand{\problemdef}[3]{
	\begin{center}
		\begin{boxedminipage}{.99\textwidth}
			\textsc{{#1}}\\[2pt]
			\begin{tabular}{ r p{0.8\textwidth}}
				\textit{~~~~Instance:} & {#2}\\
				\textit{Question:} & {#3}
			\end{tabular}
		\end{boxedminipage}
	\end{center}
}

\newcommand{\optproblemdef}[3]{
	\begin{center}
		\begin{boxedminipage}{.99\textwidth}
			\textsc{{#1}}\\[2pt]
			\begin{tabular}{ r p{0.8\textwidth}}
				\textit{~~~~Instance:} & {#2}\\
				\textit{Goal:} & {#3}
			\end{tabular}
		\end{boxedminipage}
	\end{center}
}

\newcommand{\NP}{{\sf NP}}
\newcommand{\sP}{{\sf P}}

\title{Connected Vertex Cover\\ for $(sP_1+P_5)$-Free Graphs\thanks{The authors were supported by The Leverhulme Trust (Grant RPG-2016-258). An extended abstract of the paper will appear in the proceedings of WG 2018~\cite{JPP18}.}
}
\author{Matthew Johnson
\and
Giacomo Paesani
\and
Dani\"el Paulusma}

\institute{Department of Computer Science, Durham University, UK\\
\texttt{\{matthew.johnson2,giacomo.paesani,daniel.paulusma\}@durham.ac.uk}
}

\begin{document}

\maketitle

\begin{abstract}
The {\sc Connected Vertex Cover} problem is to decide if~a graph~$G$ has a vertex cover of size at most~$k$ that induces a connected subgraph of~$G$. This is a well-studied problem, known to be \NP-complete for restricted graph classes, and, in particular, for $H$-free graphs if~$H$ is not a linear forest.
On the other hand, the problem is known to be polynomial-time solvable for $sP_2$-free graphs for any integer~$s\geq 1$. 
We give a polynomial-time algorithm to solve the problem for $(sP_1+P_5)$-free graphs for every integer~$s\geq~0$.
Our algorithm can also be used for the {\sc Weighted Connected Vertex Cover} problem.
\end{abstract}

\section{Introduction}\label{s-intro}

A set~$S$ of vertices in a graph~$G$ forms a {\it vertex cover} of~$G$ if every edge of~$G$ is incident with a vertex of~$S$.  The set $S$ is an {\it independent set} if no two vertices in~$S$ are adjacent. These definitions lead to two classical graph problems, which are both \NP-complete: the {\sc Vertex Cover} problem is to decide if a given graph~$G$ has a vertex cover of size at most~$k$ for a given integer~$k$; the {\sc Independent Set} problem is to decide if a given graph~$G$ has an independent set of size at least~$\ell$ for a given integer~$\ell$.  A set~$S$ of at least $k$ vertices of a graph $G$ on $n$ vertices is a vertex cover if and only if $V_G\setminus S$ is an independent set (of size at most~$n-k$). Hence  {\sc Vertex Cover} and {\sc Independent Set} are polynomially equivalent. A vertex cover of a graph~$G$ is connected if it induces a connected subgraph of~$G$. In our paper, we focus on the corresponding decision problem.

\problemdef{{\sc Connected Vertex Cover}}{a graph $G$ and an integer $k$.}{does $G$ have a connected vertex cover $S$ with $|S|\leq k$?} 

\noindent
In 1977, Garey and Johnson~\cite{GJ77} proved that  {\sc Connected Vertex Cover} is \NP-complete for planar graphs of maximum degree~4.  More recently, Priyadarsini and Hemalatha~\cite{PH08} and Fernau and Manlove~\cite{FM09} strengthened this result to 2-connected planar graphs of maximum degree~4 and planar bipartite graphs of maximum degree~4, respectively. Wanatabe, Kajita, and Onaga~\cite{WKO91} proved that {\sc Connected Vertex Cover} is \NP-complete even for 3-connected graphs. Very recently, Munaro~\cite{Mu17} proved the same for line graphs of planar cubic bipartite graphs and for planar bipartite graphs of arbitrarily large girth, and Li, Yang, and Wang~\cite{LYW17} showed \NP-completeness for 4-regular graphs.

We now turn to tractable cases.  Ueno, Kajitani, and Gotoh~\cite{UKG88} proved that {\sc Connected Vertex Cover} is polynomial-time solvable for graphs of maximum degree at most~3. Escoffier, Gourv\`es, and Monnot~\cite{EGM10} proved the same result for chordal graphs. As {\sc Vertex Cover} is also polynomial-time solvable for chordal graphs~\cite{Ga74}, the authors of~\cite{EGM10} proposed a general study on the complexity of {\sc Connected Vertex Cover} on graph classes for which {\sc Vertex Cover} is polynomial-time solvable. This leads us to the research question of our paper:

\medskip
\noindent
{\it For which classes of graphs do the complexities of {\sc Vertex Cover} and {\sc Connected Vertex Cover} coincide?}

\medskip
\noindent
Chiarelli, Hartinger, Johnson, Milanic, and Paulusma~\cite{CHJMP18} studied this question by considering classes of graphs characterized by a single forbidden induced subgraph~$H$. Such graphs are called $H$-free. They observed that the results of Munaro~\cite{Mu17} imply that {\sc Connected Vertex Cover} is \NP-complete for $H$-free graphs if $H$ contains a cycle or a claw. Using Poljak's construction~\cite{Po74}, {\sc Vertex Cover} is readily seen to be \NP-complete for graphs of arbitrarily large girth and thus for $H$-free graphs whenever $H$ contains a cycle. When $H$ is the claw, {\sc Vertex Cover} becomes polynomial-time solvable for $H$-free graphs~\cite{Mi80,Sh80}. Hence, there exist graphs~$H$ such that {\sc Connected Vertex Cover} and {\sc Vertex Cover} have different complexities when restricted to $H$-free graphs (assuming $\sP\neq \NP$);  see~\cite{Al04,BM18} for some more examples. 

So the complexity of {\sc Connected Vertex Cover} is known for $H$-free graphs unless $H$ is a linear forest (the disjoint union of one or more paths). Even the case where $H$ is a single path on $r$ vertices (denoted $P_r$) is settled neither for {\sc Vertex Cover} nor for {\sc Connected Vertex Cover}; it is not known if there exists an integer~$r$ such that {\sc Vertex Cover} or {\sc Connected Vertex Cover} is \NP-complete for $P_r$-free graphs. Lokshtanov, Vatshelle, and Villanger~\cite{LVV14}  proved that {\sc Independent Set}, and thus {\sc Vertex Cover}, is polynomial-time solvable for $P_5$-free graphs. Recently, Grzesik, Klimo\v{s}ov\'a, Pilipczuk, and Pilipczuk~\cite{GKPP17} extended this to $P_6$-free graphs. We also note that if {\sc Vertex Cover} is polynomial-time solvable on $H$-free graphs for some graph~$H$, then it is polynomial-time solvable on $(P_1+H)$-free graphs.  This follows from the observation (see, e.g.,~\cite{Mo12}) that to solve the complementary problem of {\sc Independent Set} on a $(P_1+H)$-free graph one  solves the problem on each $H$-free graph obtained by removing a vertex and all its  neighbours.

\begin{theorem}[\cite{GKPP17}]\label{t-vc}
For every $s\geq 0$, {\sc Vertex Cover} can be solved in polynomial time for $(sP_1+P_6)$-free graphs.
\end{theorem}

\noindent
By using the concept of the price of connectivity~\cite{CCFS14,CL10,HJMP16}, Chiarelli et al.~\cite{CHJMP18} proved that {\sc Connected Vertex Cover} is polynomial-time solvable for $sP_2$-free graphs for any integer~$s\geq 1$. For {\sc Vertex Cover} this follows by combining two classical results~\cite{BY89,TIAS77} (as is well-known). No other complexity results are known for {\sc Connected Vertex Cover} for $H$-free graphs if $H$ is a linear forest.

\subsection{Our Contribution}  We continue the study of~\cite{CHJMP18,EGM10}, and in Sections~\ref{s-poly} and~\ref{s-main2}, we prove the following result, which includes polynomial-time solvability for $P_5$-free graphs.

\begin{theorem}\label{t-main}
For every $s\geq 0$, {\sc Connected Vertex Cover} can be solved in polynomial time for $(sP_1+P_5$)-free graphs.
\end{theorem}

In fact, both Lokshtanov et al.~\cite{LVV14} and Grzesik et al.~\cite{GKPP17} showed that a more general variant of {\sc Vertex Cover} is polynomial-time solvable for $P_5$-free graphs and $P_6$-free graphs, respectively. Namely, they considered the variant, where each vertex~$v$ of the input graph has an associated non-negative weight~$w_u$ and the goal is to find a vertex cover of total minimum weight. This result can be readily extended to $(sP_1+P_6)$-free graphs by using the same observation as before. In Section~\ref{s-weight} we show how to generalize Theorem~\ref{t-main} to hold for the weighted version of {\sc Connected Vertex Cover}.

\subsection{Our Method} 
It is easy to construct graphs with a minimum connected vertex cover that do not contain a minimum vertex cover; see 
the graph $G_1$ in Fig.~\ref{f-4ex}.
We also note that the difference
in size
 between a minimum vertex cover and a minimum connected vertex cover in an $(sP_1+P_5)$-
free graph is at most $3$ if $s=0$, and at most $3s+10$ if $s\geq 1$~\cite{HJMP16}.  We cannot exploit this property directly as that would require an algorithm to enumerate all minimum vertex covers in polynomial time. Moreover, the graph $G_2$ in Fig.~\ref{f-4ex} shows that even if this were possible, it is not immediately obvious how to proceed; one cannot necessarily hope to find a minimum connected vertex cover by extending a minimum vertex cover. As an extra complication, for {\sc Connected Vertex Cover} one cannot extend results on $H$-free graphs to results on $(sP_1+H)$-free graphs in a straightforward way (certainly one cannot use the technique for {\sc Vertex Cover} described before Theorem~\ref{t-vc}).

Our method is based on an analysis of the structure of dominating sets in $(sP_1+P_5)$-free graphs using a characterization of $P_5$-free graphs due to Bacs\'{o} and Tuza~\cite{BT90}. We translate the problem into a problem in which we  try to extend a partial vertex cover into a full connected vertex cover. We solve this extension variant of {\sc Connected Vertex Cover} by using Theorem~\ref{t-vc} (applied to the smaller class of $(sP_1+P_5)$-free graphs). We show how to do this in Section~\ref{s-poly} and then show how to use this result to prove Theorem~\ref{t-main} in Section~\ref{s-main2}. 

An important ingredient of our proof is that we reduce the size of the input graph by contracting an edge between two vertices $u$ and $v$ whenever we detect that $u$ and $v$ will both belong to the connected vertex cover. This idea stems from the observation that a connected graph $G$ on $n$ vertices has a connected vertex cover of size~$k$ if and only if $G$ contains the star $K_{1,n-k}$ on $n-k+1$ vertices  as a contraction.\footnote{If $G$ has a connected vertex cover $S$ of size~$k$, then contracting every edge between vertices in $S$ modifies $G$ into $K_{1,n-k}$. If $G$ contains $K_{1,n-k}$ as a contraction, then $V_G$ can be partitioned into sets $A$, $B_1$, \ldots, $B_{n-k}$ that each induce a connected graph such that there exists at least one edge between a vertex from $A$ and a vertex from $B_i$ for $i=1,\ldots,n-k$ and no edges between two vertices from different $B$-sets.  If $|B_i|\geq 2$, then we move every vertex that is adjacent to a vertex of $A$ to $A$ until we have only one vertex in $B_i$ left. This gives us a connected vertex cover of size~$k$.} 

\usetikzlibrary{patterns}
\begin{figure}
\begin{subfigure}{.2\textwidth}
\begin{center}
\begin{tikzpicture}[scale=0.6]
\node at (-1,2) {$G_1$};
\draw (-1,-1) -- (-1,1) -- (0,2) -- (1,1) -- (1,-1) -- (0,-2) -- (-1,-1) (-1,-1) -- (1,1) (-1,1) -- (1,-1);
\draw [fill=white] (-1,1) circle [radius=5pt]
			(1,1) circle [radius=5pt]
			(0,-2) circle [radius=5pt];
\draw [fill=black] (0,2) circle [radius=5pt]
		      (-1,-1) circle [radius=5pt]
		      (1,-1) circle [radius=5pt];
\end{tikzpicture}
\end{center}
\end{subfigure}%
\begin{subfigure}{.3\textwidth}
\begin{center}
\begin{tikzpicture}[scale=0.6]	
\node at (-1,2) {$G_1$};

\draw (-1,-1) -- (-1,1) -- (0,2) -- (1,1) -- (1,-1) -- (0,-2) -- (-1,-1);
\draw (-1,-1) -- (1,1) (-1,1) -- (1,-1);
\draw [fill=white] (0,2) circle [radius=5pt]
			(0,-2) circle [radius=5pt];
\draw [fill=black] (-1,-1) circle [radius=5pt]
		      (1,-1) circle [radius=5pt]
			(1,1) circle [radius=5pt]
			(-1,1) circle [radius=5pt];
\end{tikzpicture}
\end{center}
\end{subfigure}%
\begin{subfigure}{.2\textwidth}
\begin{center}
\begin{tikzpicture}[scale=0.6]
\node at (-1.8,2) {$G_2$};
\draw  (1,2) -- (-1,2) -- (-1,0) -- (-1,-2) --  (-1,-2) -- (1,-2) -- (1,0) -- (-1,0);
\draw [fill=white] (-1,2) circle [radius=5pt]
			(1,0) circle [radius=5pt]
			(-1,-2) circle [radius=5pt];
\draw [fill=black] (1,2) circle [radius=5pt]
		      (-1,0) circle [radius=5pt]
		      (1,-2) circle [radius=5pt];
\end{tikzpicture}
\end{center}
\end{subfigure}%
\begin{subfigure}{.3\textwidth}
\begin{center}
\begin{tikzpicture}[scale=0.6]
\node at (-1.8,2) {$G_2$};
\draw  (-1,2) -- (1,2) (-1,-2) -- (1,-2) -- (1,0);
\draw (-1,2) -- (-1,0) -- (-1,-2) (-1,0) -- (1,0);
\draw [fill=white] (1,-2) circle [radius=5pt]
			(1,2) circle [radius=5pt];
\draw [fill=black] (-1,-2) circle [radius=5pt]
		      (-1,0) circle [radius=5pt]
		      (-1,2) circle [radius=5pt]
		      (1,0) circle [radius=5pt];
\end{tikzpicture}
\end{center}
\end{subfigure}
\caption{An example of a $P_5$-free graph~$G_1$ with a minimum connected vertex cover (coloured black in the right-hand drawing) that contains no minimum vertex cover (there are exactly two, indicated by the sets of black and white vertices in the left-hand drawing). The graph $G_2$ is  an example of a $(P_1+P_5)$-free graph with a minimum vertex cover (coloured black in the left hand drawing) that is not contained in any minimum connected vertex cover; clearly any connected vertex cover that contains it has at least five vertices and an example of a minimum connected vertex cover on four vertices is indicated by the vertices coloured black in the right-hand drawing.}\label{f-4ex}
\end{figure}
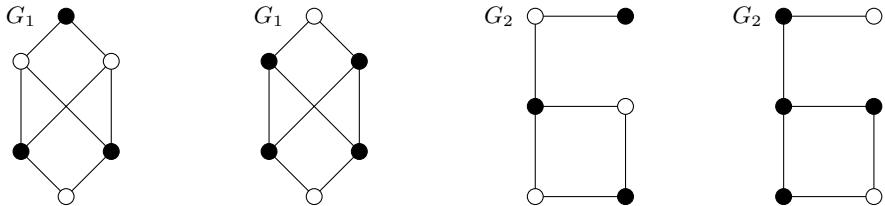

\subsection{Related Work on $(P_1+P_r)$-Free Graphs and $P_r$-Free Graphs}

The class of $P_5$-free graphs has also been studied for other problems than {\sc Vertex Cover} and {\sc Connected Vertex Cover}. In fact the computational complexity of many of these problems jumps from polynomial-time solvable on $P_r$-free graphs to \NP-complete on $(P_1+P_r,P_{r+1})$-free graphs.
For instance, {\sc Colouring} is  polynomial-time solvable for $P_4$-free graphs but is \NP-complete for 
$(P_1+P_4,P_5)$-free graphs~\cite{KKTW01}. 
Later, Ho{\`{a}}ng et~al.~\cite{HKLSS10} proved that {\sc $k$-Colouring} is polynomial-time solvable for $P_5$-free graphs for every $k\geq 1$. Afterwards, this result was extended to $(sP_1+P_5)$-free graphs for any $s\geq 0$~\cite{CGKP15}. 
A clique transversal of a graph $G$ is a set $S\subseteq V_G$ such that $S$ contains a vertex of each maximal clique of $G$ (note that a vertex cover can be viewed as a transversal which contains a vertex of each 2-vertex clique).  It is known that computing a smallest clique transversal can be done in polynomial time for comparability graphs~\cite{BNP96} and thus for $P_4$-free graphs, but is \NP-hard for cobipartite graphs~\cite{GP00} and thus for $(P_1+P_4,P_5)$-free graphs.
The {\sc Longest Path Contractibility}~\cite{HPW09} is to determine the length of a longest path to which a graph can  be contracted. This problem is polynomial-time solvable for $(P_1+P_5)$-free graphs~\cite{KP} but \NP-hard for $P_6$-free graphs
~\cite{HPW09}.
Golovach and Heggernes~\cite{GH09} gave a fixed-parameter tractable algorithm for {\sc Choosability} on $P_5$-free graphs when parameterized by the size of the lists of admissible colours. Recently, Bonamy et al.~\cite{BDFJP17} proved that the problems {\sc Independent Feedback Vertex Set} and {\sc Independent Odd Cycle Transversal} are polynomial-time solvable for $P_5$-free graphs.

\section{Preliminaries}\label{s-pre}

We consider only finite, undirected graphs without multiple edges or self-loops. Let $G=(V,E)$ be a graph. We let $n=|V|$. For a set $S\subseteq V$, the graph $G[S]$ denotes the subgraph of $G$ induced by $S$, and we say that $S$ is {\it connected} if~$G[S]$ is connected. We write $G-S=G[V\setminus S]$, and if $S=\{u\}$ we may simply write $G-u$. For a vertex $u\in V$, we write $N_G(u)=\{v \;|\; uv\in E\}$ to denote the neighbourhood of~$u$. For a set $S\subseteq V$, we write $N_G(S)=(\bigcup_{u\in S}N_G(u)) \setminus S$.  A subset~$D\subseteq V$ is a {\it dominating} set of~$G$ if every vertex of $V\setminus D$ is adjacent to at least one vertex of $D$.  An edge~$uv$ of a graph $G=(V,E)$ is {\it dominating} if $\{u,v\}$ is dominating. The \emph{contraction} of an edge $uv\in E$ is the operation that replaces $u$ and $v$ by a new vertex adjacent to precisely those vertices of $V \setminus \{u,v\}$ adjacent to $u$ or $v$ in $G$.  Recall that for a graph~$H$, we say that another graph~$G$ is {\it $H$-free} if it does not contain an induced subgraph isomorphic to $H$. The {\em disjoint union} $G+\nobreak H$ of two vertex-disjoint graphs~$G$ and~$H$ is the graph $(V_G\cup V_H, E_G\cup E_H)$. The disjoint union of~$r$ copies of a graph~$G$ is denoted by~$rG$. A  {\em linear forest} is the disjoint union of one or more paths. The following, straightforward lemma holds for any linear forest, but,  as we repeatedly make use of it, it is convenient to state in these terms.

\begin{lemma}\label{l-contract}
Let $G$ be a connected $(sP_1+P_5)$-free graph for some $s\geq 0$. The graph obtained from $G$ after contracting an edge is also connected and $(sP_1+P_5)$-free.
\end{lemma}

We will use the following result of Bacs\'{o} and Tuza~\cite{BT90} as a lemma in our proof. 

\begin{lemma}[\cite{BT90}]\label{l-bt}
Every connected $P_5$-free graph $G$ has a dominating set~$D$, computable in $O(n^3)$ time, that induces either a $P_3$ or a complete graph. 
\end{lemma}

\noindent
Note that it is not difficult to compute the set~$D$ in $O(n^3)$ time; this also follows from a more general result of Camby and Schaudt~\cite{CS16} for $P_r$-free graphs ($r\geq 1$).

\section{An Auxiliary Problem}\label{s-poly}

In this section we prove that a variant of {\sc Connected Vertex Cover} can be solved in polynomial time for $(sP_1+P_5)$-free graphs for every integer~$s\geq 0$. To prove Theorem~\ref{t-main} we will solve a polynomial number of instances of this variant, which we show can be solved in polynomial time for $(sP_1+P_5)$-free graphs for every $s\geq 0$.  We introduce the variant by first describing its input. Let $G$ be a connected graph, let $J\subseteq V_G$ be a subset of the vertex set of $G$ and let $y$ be a vertex of $J$. We call the triple $(G,J,y)$ {\it cover-complete} if it has the following properties (see also Fig.~\ref{f-ccc}):
\begin{itemize}
\item [(A)] $J$ is an independent set; 
\item [(B)] $y$ is adjacent to every vertex of $G-J$;
\item [(C)] the neighbours of each vertex in $J\setminus \{y\}$ form an independent set in $G-J$.
\end{itemize}
We now describe the problem.

\optproblemdef{{\sc Connected Vertex Cover Completion}}{a cover-complete triple $(G,J,y)$.}{find a smallest connected vertex cover $S$ of $G$ such that $J\subseteq S$.} 

\noindent
We will show how to solve  this problem in polynomial time for $(sP_1+P_5)$-free graphs for any $s\geq 0$. We first make some further definitions and then prove a number of lemmas.

Let $(G,J,y)$ be a cover-complete triple, where $G$ is a connected $(sP_1+P_5)$-free graph. For a vertex  $w\in N_G(J\setminus \{y\})$, we write $J_w= N_G(w)\cap J$. Note that, by (B), $y\in J_w$. Let $G'$ be the graph obtained from $G$ by contracting every edge of $G[J_w\cup \{w\}]$.  As $G[J_w\cup \{w\}]$ is connected, contracting its edges reduces it to a single vertex which we denote $y_w$. We say that we have {\it set-contracted} $G$ into $G'$ via $w$ and that we {\it contracted} $J_w\cup \{w\}$ into $y_w$;
see Fig.~\ref{f-ccc} for an example. 

\begin{figure}
\begin{subfigure}{.5\textwidth}
\begin{center}
\begin{tikzpicture}[scale=0.5]
\draw [gray!20!white, fill=gray!20!white] (-5.5,-4) rectangle (5.5,-2);
\draw [gray!20!white, fill=gray!20!white] (-5.5,2) rectangle (5.5,4);
\draw [black, dotted] (-5.25,-3.75) rectangle (0.25,-2.25);
\node [left] at (-5.5,3) {$L$};
\node [left] at (-5.5,-3) {$J$};
\filldraw [black] (-5,-3) circle [radius=3pt] 
		(-2.5,-3) circle [radius=3pt]
		(0,-3) circle [radius=3pt]
		(2.5,-3) circle [radius=3pt]
		(5,-3) circle [radius=3pt]
		(-3,5) circle [radius=3pt]
		(-5,3) circle [radius=3pt]
		(-1,3) circle [radius=3pt]
		(1,3) circle [radius=3pt]
		(3,3) circle [radius=3pt]
		(5,3) circle [radius=3pt];		
\draw [black] (-5,-3) -- (-3,5) (-5,-3) -- (-5,3) (-5,-3) -- (-1,3) (-5,-3) -- (1,3) (-5,-3) -- (3,3) (-5,-3) -- (5,3) (3,3) -- (5,3) (-5,3) -- (-3,5)
		(-5,3) -- (-2.5,-3) (-5,3) -- (0,-3) (-1,3) -- (-2.5,-3) (-1,3) -- (2.5,-3) (1,3) -- (-2.5,-3) (1,3) -- (0,-3) (1,3) -- (5,-3) (3,3) -- (-2.5,-3) (5,3) -- (0,-3) (5,3) -- (2.5,-3);
\node [below right] at (-5,-3) {$y$};
\node [above] at (-5,3) {$w$};
\node at (-1,-3) {$J_w$};
\end{tikzpicture}
\end{center}
\end{subfigure}%
\begin{subfigure}{.5\textwidth}
\begin{center}
\begin{tikzpicture}[scale=0.5]
\draw [gray!20!white, fill=gray!20!white] (-5.5,-4) rectangle (5.5,-2);
\draw [gray!20!white, fill=gray!20!white] (-5.5,2) rectangle (5.5,4);
\node [left] at (6.5,3) {$L'$};
\node [left] at (6.5,-3) {$J'$};
\filldraw [black] (-2.5,-3) circle [radius=3pt]
		(2.5,-3) circle [radius=3pt]
		(5,-3) circle [radius=3pt]
		(-3,5) circle [radius=3pt]
		(-1,3) circle [radius=3pt]
		(1,3) circle [radius=3pt]
		(3,5) circle [radius=3pt]
		(5,3) circle [radius=3pt];
\draw [black] (-2.5,-3) -- (-3,5) (-2.5,-3) -- (-1,3) (-2.5,-3) -- (1,3) (-2.5,-3) -- (3,5) (-2.5,-3) -- (5,3) (-1,3) -- (2.5,-3) -- (5,3) -- (3,5) (1,3) -- (5,-3);
\node [below] at (-2.5,-3) {$y_w$};
\end{tikzpicture}
\end{center}
\end{subfigure}
\caption{An example of a cover-complete triple $(G,J,y)$ and the cover-complete triple $(G',J',y_w)$ obtained from set-contracting~$G$ via vertex~$w$. 
The sets $J'=(J\setminus J_w)\cup \{y_w\}$, $L=N_G(J\setminus \{y\})$ and $L'=N_{G'}(J'\setminus \{y_w\})$ are also displayed (the latter two sets will be formally introduced later).}\label{f-ccc}
\end{figure}
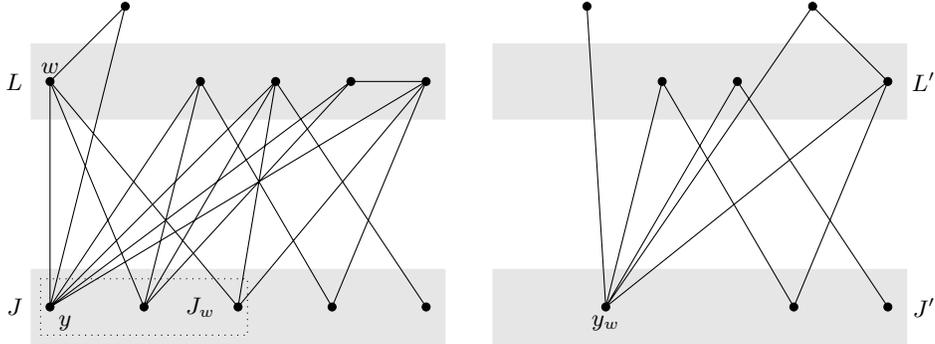

The following lemma is crucial.

\begin{lemma}\label{l-crucial}
Let $(G,J,y)$ be a cover-complete triple, where $G$ is a connected $(sP_1+P_5)$-free graph for some $s\geq 0$. Let $w\in N_G(J\setminus \{y\})$, and let $G'$ be the graph obtained from $G$ after set-contracting via~$w$. Let $J'=(J\setminus J_w)\cup \{y_w\}$ and $y'=y_w$. Then the following statements hold:
\begin{itemize}
\item [1.] $G'$ is a connected $(sP_1+P_5)$-free graph;\\[-10pt]
\item [2.] $(G',J',y')$ is a  cover-complete triple;\\[-10pt]
\item [3.] A set $S\subseteq V_G$ is a (smallest) connected vertex cover of $G$ that contains $J\cup \{w\}$ if and only if $(S\setminus (J\cup \{w\}))\cup J'$ is a (smallest) connected vertex cover of $G'$ that contains $J'$.
\end{itemize}
\end{lemma}

\begin{proof}
We will prove 1-3 separately.

\medskip
\noindent
{\emph 1.} By Lemma~\ref{l-contract}, $G'$ is connected and $(sP_1+P_5)$-free. This proves 1.

\medskip
\noindent
{\emph 2.} We will prove (A)-(C) for $(G',J',y')$. Before we do this we first observe the following. As (B) holds for $(G,J,y)$, we find that $y\in J$ is adjacent to $w$ in $G$. Hence $y$ belongs to $J_w$ and thus to $J_w\cup \{w\}$, which is contracted to the single vertex~$y'$ in $G'$. Hence, $y$ is not in $G'$ and its role has been taken over by $y'$, as we show below.

We first prove (A). As $J$ is an independent  set in $G$, we find that $J\setminus J_w$ is an independent set in $G'$. For contradiction, suppose that $y'$ is adjacent to a vertex in $J\setminus J_w$. Then there is an edge between a vertex of  $J\setminus J_w$ and a vertex of $J_w\cup \{w\}$ in $G$. However, this not possible as $J$ is independent in $G$, and thus every edge in $G[J\cup \{w\}]$ is incident with~$w$. Hence $J'=(J\setminus J_w) \cup \{y'\}$ is an independent set in $G'$. This proves (A).

We now prove (B). Recall that $y$ belongs to $J_w\cup \{w\}$, which is contracted to~$y'$ in $G'$. Hence, as $y$ is adjacent to every vertex of $G-J$ in $G$, we find that $y'$ is adjacent to every vertex of $G'-J'$. This proves (B).

Finally we prove (C). Let $x\in J'\setminus \{y'\}$. Then $x$ is not adjacent to $y'$, as we showed above that $J'$ is an independent set in $G'$. Then $N_{G'}(x)=N_G(x)$ is an independent set, as (C) holds for $(G,J,y)$. This proves (C) and 2.

\medskip
\noindent
{\emph 3.} Any connected vertex cover $S$ of $G$ that contains $J\cup \{w\}$ contains every vertex of $J_w\cup \{w\}$. Hence contracting $J_w\cup \{w\}$ to $y'$ yields a connected vertex cover $(S\setminus (J\cup \{w\}))\cup J'$ of~$G'$ that contains $J'$. Any connected vertex cover $S'$ of $G'$ that contains $J'$ contains $y'$. Hence uncontracting the edges of $G[J_w\cup \{w\}]$ yields a connected vertex cover $(S'\cup J\cup \{w\})\setminus J'$ of $G$ that contains $J\cup \{w\}$. This proves~3.\qed
\end{proof}

Let $(G,J,y)$ be a  cover-complete triple. We define $L_J=N_G(J\setminus \{y\})$.  If there is no ambiguity, we will just write $L=L_J$
(see also Fig.~\ref{f-ccc}).
Note that, by (C), $L$ is the union of a number of independent sets, but $L$ itself might not be independent. However we can deduce the following lemma, which follows immediately from~(C).

\begin{lemma}\label{l-ind}
Let $(G,J,y)$ be a  cover-complete triple. If $w_1$ and $w_2$ are two adjacent vertices in $L$, then no vertex of $J\setminus \{y\}$ is adjacent to both $w_1$ and $w_2$.
\end{lemma}

We introduce two key definitions for a cover-complete triple $(G,J,y)$.  Two vertices $w_1,w_2\in L$ form a {\it pseudo-dominating pair}~if 
\begin{itemize}
\item $w_1$ and $w_2$ are non-adjacent; 
\item $w_1$ has a neighbour $x_1\in J$ not adjacent to $w_2$; and 
\item $w_2$ has a neighbour $x_2\in J$ not adjacent to $w_1$. 
\end{itemize}
\noindent
Three  vertices $w_1,w_2,w_3\in L$ form a {\it pseudo-dominating triple} if
\begin{itemize}
\item $w_1$ is adjacent to neither $w_2$ nor $w_3$; 
\item $w_2$ and $w_3$ are adjacent; 
\item $J$ contains two distinct vertices $x_1$ and $x_2$ such that 
\begin{itemize}
\item $x_1\in N_G(w_1)\setminus N_G(\{w_2,w_3\})$ and 
\item $x_2\in (N_G(w_1)\cap N_G(w_2))\setminus N_G(w_3)$.
\end{itemize}
\end{itemize}
See the illustrations in Fig.~\ref{f-pair}, from which we also observe that no pseudo-dominating pair or pseudo-dominating triple can be found in a $P_5$-free graph.
\begin{figure}[h]
\begin{subfigure}{.5\textwidth}
\begin{center}
\begin{tikzpicture}[scale=1]

\draw [gray!20!white, fill=gray!20!white] (-2.5,0) rectangle (2,1);
\node [left] at (-2.5,0.5) {$L$};
\draw [gray!20!white, fill=gray!20!white] (-2.5,-2.5) rectangle (2,-1.5);
\node [left] at (-2.5,-2) {$J$};
\filldraw [black] (-0.5,0.5) circle [radius=2pt] 
			   (1,0.5) circle [radius=2pt]
			   (-0.5,-1.75) circle [radius=2pt]
		         (1,-1.75) circle [radius=2pt]
		         (-2,-2.25) circle [radius=2pt];
\draw [very thick] (1,0.5) -- (1,-1.75) (-0.5,0.5) -- (-0.5,-1.75) (-0.5,0.5) -- (-2,-2.25) -- (1,0.5);
\draw [dashed] (-0.5,0.5) -- (1,-1.75) -- (-0.5,-1.75) -- (1,0.5) -- (-0.5,0.5) (-0.5,-1.75) -- (-2,-2.25) -- (1,-1.75);
\node [above left] at (-0.5,0.5) {$w_1$};
\node [above right] at (1,0.5) {$w_2$};
\node [above left] at (-0.5,-1.75) {$x_1$};
\node [above right] at (1,-1.75) {$x_2$};
\node [above left] at (-2,-2.25) {$y$};
\end{tikzpicture}
\end{center}
\end{subfigure}
\begin{subfigure}{.5\textwidth}
\begin{center}

\begin{tikzpicture}
\draw [gray!20!white, fill=gray!20!white] (-2.5,0) rectangle (2,1);
\node [left] at (2.4,0.5) {$L$};

\draw [gray!20!white, fill=gray!20!white] (-2.5,-2.5) rectangle (2,-1.5);
\node [left] at (2.4,-2) {$J$};

\filldraw [black] (-2,0.75) circle [radius=2pt]
			    (-0.5,0.25) circle [radius=2pt]
		          (1,0.75) circle [radius=2pt]
		          (-2,-2) circle [radius=2pt]
			    (-0.5,-2) circle [radius=2pt];
\draw [very thick] (-2,-2) -- (-2,0.75) -- (-0.5,-2) -- (-0.5,0.25) -- (1,0.75);
\draw [dashed] (-2,0.75) -- (1,0.75) -- (-0.5,-2) -- (-2,-2) -- (-0.5,0.25) -- (-2,0.75);
\draw [dashed] (-2,-2) -- (1,0.75);
\node [above left] at (-2,0.65) {$w_1$};
\node [above] at (-0.5,0.25) {$w_2$};
\node [above right] at (1,0.65) {$w_3$};
\node [below left] at (-2,-2) {$x_1$};
\node [below right] at (-0.5,-2) {$x_2$};
\end{tikzpicture}
\end{center}
\end{subfigure}
\caption{Examples, on the left, of a pseudo-dominating pair $(w_1,w_2)$, and, on the right, of a pseudo-dominating triple $(w_1,w_2,w_3)$. As easily seen, the presence of either implies the existence of at least one induced $P_5$.}\label{f-pair}
\end{figure}
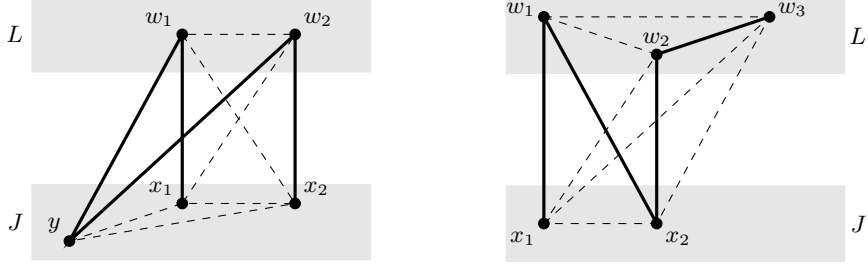

Let $S$ be a connected vertex cover of $G$ that contains $J$. Recall that $J$ is an independent set.  A subset $L^*\subseteq L\cap S$ is a {\it connector} of $S$ if $J\cup L^*$ is connected. We present the following two lemmas. 

\begin{lemma}\label{l-pseudopair}
Let $(G,J,y)$ be a  cover-complete triple, where $G$ is an $(sP_1+P_5)$-free graph for some $s\geq 0$. Let $S$ be a connected vertex cover of $G$ that contains $J$. If $S$ contains both vertices of a pseudo-dominating pair $w_1$, $w_2$, then $S$~has a connector of size at most $s+1$ that contains both $w_1$ and $w_2$.
\end{lemma}

\begin{proof}
By definition, there exist two vertices $x_1$ and $x_2$ in $J$, such that $w_1$ is not adjacent to $x_2$ and $w_2$ is not adjacent to $x_1$. As $J$ is an independent set by (A) and each vertex of $L$ is adjacent to $y$ by (B), we find that  $\{x_1,w_1,y,w_2,x_2\}$ induces a $P_5$ in that order. As $G$ is $(sP_1+P_5)$-free and $J$ is an independent set, this means that $\{w_1,w_2\}$ dominates all vertices of $J$ except for a subset $I\subseteq J$ of at most $s-1$ vertices. We choose $L^*$ to consist of $w_1$, $w_2$ and a neighbour in $L\cap S$ of each vertex of $I$ (note that such a neighbour must exist for each vertex of $I$ as $S$ is connected). Then  $J\cup L^*$ is connected, that is, $L^*$ is a connector, as each vertex of $J$ is adjacent to some vertex of $L^*$ and each vertex of $L^*$ is adjacent to $y\in J$ due to (B). Moreover, $L^*$ has size at most $s+1$.\qed
\end{proof}

\begin{lemma}\label{l-pseudotriple}
Let $(G,J,y)$ be a  cover-complete triple, where $G$ is an $(sP_1+P_5)$-free graph for some $s\geq 0$. Let $S$ be a connected vertex cover of $G$ that contains~$J$. If $S$ contains all three vertices of a pseudo-dominating triple $w_1,w_2,w_3$, then $S$ has a connector of size at most $s+2$ that contains $\{w_1,w_2,w_3\}$.
\end{lemma}

\begin{proof}
By definition, there exist two vertices $x_1$ and $x_2$ in $J$ such that $x_1$ is adjacent to $w_1$ but not to $w_2$ and $w_3$, and $x_2$ is adjacent to $w_1$ and $w_2$ but not $w_3$. Then $\{x_1,w_1,x_2,w_2,w_3\}$ induce a $P_5$ in that order. As $G$ is $(sP_1+P_5)$-free and $J$ is an independent set, this means that $\{w_1,w_2,w_3\}$ dominates all vertices of $J$ except for a subset $I\subseteq J$ of at most $s-1$ vertices. We choose $L^*$ to consist of $w_1$, $w_2$, $w_3$ and a neighbour in $L\cap S$ of each vertex of $I$ (note that such a neighbour must exist for each vertex of $I$ as $S$ is connected). Then  $J\cup L^*$ is connected, that is, $L^*$ is a connector, as each vertex of $J$ is adjacent to some vertex of $L^*$ and each vertex of $L^*$ is adjacent to $y\in J$ due to (B). Moreover, $L^*$ has size at most $s+2$.\qed
\end{proof}

Let $(G,J,y)$ be a  cover-complete triple. Let $S$ be a connected vertex cover of $G$ that contains $J$. If $S$ contains both vertices of some pseudo-dominating pair of $G$ or all three vertices of some pseudo-dominating triple of $G$, then $S$ is of {\it type~1}. Otherwise $S$ must contain at most one vertex of any pseudo-dominating pair and at most two vertices of any pseudo-dominating triple of $G$. In that case we say that $S$ is of {\it type~2}. We observe that $G$ might have connected vertex covers of only one type.

We will now see, in Lemma~\ref{l-first}, how to find a smallest type~1 connected vertex cover of a graph $G$ of a  cover-complete triple $(G,J,y)$ in polynomial time (if it exists). After that we shall prove how to find a smallest type~2 connected vertex cover of $G$ in polynomial time (if it exists). To compute these sets we need the following lemma, which uses Theorem~\ref{t-vc} in its  proof. 

\begin{lemma}\label{l-vcvc}
Let $(G,\{y\},y)$ be a  cover-complete triple, where $G$ is an $(sP_1+P_5)$-free graph for some $s\geq 0$. Then it is possible to compute a smallest connected vertex cover of $G$ that contains $y$ in $O(n^{s+14})$ time.
\end{lemma}

\begin{proof}
As $(G,\{y\},y)$ is a  cover-complete triple, $y$ dominates $G$. Moreover $G-y$ is $(sP_1+P_5)$-free. Then we can compute a smallest  vertex cover~$S$ of $G-y$ by using Theorem~\ref{t-vc}.  As $y$ dominates $G$, we find that $S\cup \{y\}$ is a smallest connected vertex cover of~$G$ that contains $y$. This takes $O(n^{s+14})$ time, which can be seen by considering the complementary problem of computing a maximum independent set in $G-y$. 
We first check by brute force, in time $O(n^s)$, if the size of a maximum independent set in $G-y$ has size at most $s$. If so, then we are done. Otherwise we consider each possibility of choosing a set~$S$ of $s$ independent vertices of $G-y$ to be in the independent set of $G-y$. For each choice $S$, we remove all vertices of $S$ and their neighbours from $G-y$. The remaining graph is $P_5$-free and we then use the algorithm of~\cite{LVV14}, which runs in $O(n^{14})$ time, to find a maximum independent set in it. \qed
\end{proof} 

Using Lemmas~\ref{l-pseudopair}--\ref{l-vcvc}, we are now ready to deal with type~1 smallest connected vertex covers.

\begin{lemma}\label{l-first}
Let $(G,J,y)$ be a  cover-complete triple. It is possible to find in $O(n^{2s+16})$ time a smallest type~1 connected vertex cover of~$G$.
\end{lemma}

\begin{proof}
We can compute all pseudo-dominating pairs of $G$ by examining each pair of vertices in turn.  This takes $O(n)$ time per pair. As
the number of pseudo-dominating pairs is $O(n^2)$, this takes $O(n^3)$ time in total.

For each pseudo-dominating pair $(w_1,w_2)$ of $G$, we describe how to compute a smallest connected vertex cover $S_{w_1,w_2}$ of $G$ that contains $J\cup \{w_1,w_2\}$. 
By Lemma~\ref{l-pseudopair}, such a vertex cover must have a connector~$L^*$ of size at most $s+1$ that contains $w_1$ and $w_2$. We find each such connector~$L^*$ by considering all sets of up to $s-1$ vertices and asking whether, combined with $w_1$ and $w_2$, they form such a connector.

For each such set~$L^*$, we do as follows. We first check if $J\cup L^*$ is connected. If so, then we apply Lemma~\ref{l-crucial} recursively for each $w\in L^*$. This takes $O(n^2)$ time, as we can use Breadth First Search and set contract at the same time.
Let $(G',J',y')$ be the resulting  cover-complete triple. Then $J'=\{y'\}$, which means  we can apply Lemma~\ref{l-vcvc} to find a smallest connected vertex cover $S'$ of $G'$ in $O(n^{14+s})$ time. 
By Lemma~\ref{l-crucial}, we can translate $S'$ into the desired vertex cover $S_{w_1,w_2}$ by uncontracting any contracted edges. As, for each pseudo-dominating pair, the number of sets $L^*$ that contain them is $O(n^{s-1})$, and the number of pseudo-dominating pairs is $O(n^2)$, the time needed to find these vertex covers is $O(n^{2s+15})$.

For each pseudo-dominating triple $(w_1,w_2,w_3)$ of $G$ we compute a smallest connected vertex cover~$S_{w_1,w_2,w_3}$ of $G$ that contains $J\cup \{w_1,w_2,w_3\}$. We can do this in  
$O(n^{2s+16})$  time by exactly the same arguments: the only differences are that the number of pseudo-dominating triples is $O(n^3)$ and that we need to apply Lemma~\ref{l-pseudotriple} instead of Lemma~\ref{l-pseudopair}.

From all the computed sets $S_{w_1,w_2}$ and $S_{w_1,w_2,w_3}$ we keep track (in constant time) of a smallest one, and in the end this yields a smallest type~1 connected vertex cover of~$G$. This proves Lemma~\ref{l-first}.\qed
\end{proof}

Let $(G,J,y)$ be a  cover-complete triple.  Using Lemma~\ref{l-first} we can find a smallest type~1 connected vertex cover of $G$ in polynomial time.  However, it might be possible that $G$ has a smaller connected vertex cover of type~2.  To investigate this, we introduce two reduction rules that will transform a  cover-complete triple $(G,J,y)$ into a triple $(G',J',y')$ with $|J'|<|J|$. We say that such a rule is {\it safe} if the following three conditions hold: 

\begin{itemize}
\item [1.] If $G$ is $(sP_1+P_5)$-free and connected, then $G'$ is $(sP_1+P_5)$-free and connected.
\item [2.] $(G',J',y')$ is  cover-complete.
\item [3.] Given a smallest connected vertex cover $S'$ of $G'$ that contains $J'$, it is possible, in $O(n^{2s+16})$ time, to find a smallest connected vertex cover~$S$ of $G$ that contains~$J$.
\end{itemize}

\noindent
{\bf Rule 1.} Set-contract via $x$ whenever $x$ is a vertex in $L\cap N_{G}(w_1)\cap N_{G}(w_2)$ for some pseudo-dominating pair $(w_1,w_2)$.

\medskip

\noindent
{\bf Rule 2.} For any vertex $w_5\in L$ that is not adjacent to any vertex of a clique of four vertices $w_1,w_2,w_3,w_4$ in $L$, delete $w_5$ and set-contract via $u$ for every $u\in L\cap N_G(w_5)$.

\begin{lemma}\label{l-rules}
Rules~1 and 2 are safe.
\end{lemma}

\begin{proof}
We first consider Rule 1.

Let $(G',J',y')$ be the resulting triple after an application of Rule~1, where $J'=(J\setminus J_x)\cup \{y_x\}$ and $y'=y_x$. By Lemma~\ref{l-crucial}, $(G',J',y')$ is  a  cover-complete triple. By the same lemma, $G'$ is $(sP_1+P_5)$-free and connected if $G$ is $(sP_1+P_5)$-free and connected. Hence we have proven that conditions~1 and~2 hold.

We are left to prove condition~3.  Let $S'$ be a smallest connected vertex cover in $G'$ that contains $J'$.  Then $S=(S'\setminus \{y'\})\cup J_x\cup \{x\}$ is a smallest connected vertex cover of $G$ that contains $J\cup \{x\}$ due to Lemma~\ref{l-crucial}. We prove the following claim.

\medskip
\noindent
{\it Claim~1. For any type~2 connected vertex cover $T$ of $G$, it holds that $|T|\geq |S|$.}

\medskip
\noindent
We prove Claim~1 as follows. Let $T$ be a connected vertex cover $T$ of $G$ that is of type~2. Suppose $x\notin T$. Then, as $x$ is adjacent to both $w_1$ and $w_2$, we find that $T$ contains both $w_1$ and $w_2$. Thus $T$ is not of type~2, a contradiction. Hence $T$ contains $x$. This implies that the set $T'=(T\setminus (J\cup \{x\}))\cup J'$ is a connected vertex cover of $G'$ that contains $J'$. As $S'$ is a smallest connected vertex cover of $G'$ that contains $J'$, we find that $|T'|\geq |S'|$. Hence $|T| = |T'|+|J_x| \geq |S'|+|J_x|=|S|$. This proves Claim~1.

\medskip
\noindent
The above means that we can do as follows. 
Given $S'$ we compute $S=(S'\setminus \{y'\})\cup J_x\cup \{x\}$ in constant time.  By Lemma~\ref{l-first} we can also compute, in $O(n^{2s+16})$ time, a smallest type~1 connected vertex cover~$S^*$ of $G$ (note that $S=S^*$ is possible). If $S$ is of type~2, then $S$ is a smallest type~2 connected vertex cover of $G$, due to Claim~1. We compare $|S|$ and $|S^*|$ and choose the smallest one. If $S$ is of type~1, then $S^*$ is a smallest connected vertex cover of $G$, again due to Claim~1.  This proves condition~3 and completes the proof that Rule 1 is safe.

\medskip
\noindent
We now consider Rule 2.
We first show that $w_5$ cannot be in any connected vertex cover~$S$ of $G$ that is of type~2. For contradiction, suppose that $w_5$ is in such a connected vertex cover $S$. Because $S$ is a vertex cover and $\{w_1,w_2,w_3,w_4\}$ is a clique, $S$ contains at least three of $\{w_1,w_2,w_3,w_4\}$, say $w_1$, $w_2$, $w_3$.

For $i=1,\ldots,5$, let $X_i$ be the set of neighbours of $w_i$ in $J$. As $w_i\in L$, every $X_i\neq \emptyset$ by definition of $L$. By Lemma~\ref{l-ind}, we find that $X_1\cap X_2\cap X_3=\emptyset$. Let $x\in X_1$. If $x\notin X_5$, then $X_5 \subseteq X_1$, as otherwise $(w_1,w_5)$ is a pseudo-dominating pair of vertices that are both contained in~$S$, which is not possible as $S$ is of type~2. As $X_1\cap X_2=\emptyset$, we find that $X_5\cap X_2=\emptyset$. This means that $(w_2,w_5)$ is a pseudo-dominating pair of vertices that are both contained in~$S$, which is not possible either. Hence $x\in X_5$. We conclude that $X_1\subseteq X_5$. For the same reason, we find that $X_2\subseteq X_5$ and $X_3\subseteq X_5$. 

Recall that $X_1\cap X_2\cap X_3=\emptyset$. Hence we can pick a vertex $x_1\in X_1$ and a vertex $x_3\in X_3$, which are both adjacent to $w_5$ but not to $w_2$, and so find that $(w_5,w_1,w_2)$ is a pseudo-dominating triple. As all three vertices $w_1$, $w_2$, $w_5$ belong to $S$, while $S$ is of type~2, this is not possible. Hence $S$ does not contain~$w_5$.

As no connected vertex cover of $G$ of type~2 may contain~$w_5$, any connected vertex cover of $G$ that is of type~2 must contain all neighbours of $w_5$, and we can delete $w_5$. The proof of conditions~1--3 is identical to the proof for Rule 1 where the neighbours of $w_5$ in $L$ take the role of the vertex~$x$ in the proof for Rule 1. \qed
\end{proof}

We call a cover-complete triple $(G,J,y)$ {\it free} if $G$ has no pseudo-dominating pair with a common neighbour in $L$, and moreover, $G[L]$ is $(P_1+K_4)$-free. By exhaustively applying Rules~1 and~2 in arbitrary order, which we may safely do due to Lemma~\ref{l-rules}, we have the following lemma.

\begin{lemma}\label{l-eventual}
A  cover-complete triple $(G,J,y)$ can be modified, in $O(n^6)$ time, into a free cover-complete triple $(G',J',y')$ 
with the following \mbox{properties}:
\begin{itemize}
\item [1.] If $G$ is $(sP_1+P_5)$-free and connected, then $G'$ is $(sP_1+P_5)$-free and connected.\\[-10pt].
\item [2.] Given a smallest connected vertex cover $S'$ of $G'$ that contains $J'$, it is possible to find in $O(n^{2s+17})$ time a smallest connected vertex cover~$S$ of $G$ that contains $J$.
\end{itemize}
\end{lemma}

\begin{proof}
We exhaustively apply Rules~1 and~2 in arbitrary order. Checking if Rule~1 can be applied takes $O(n^3)$ time, as there are $O(n^2)$ pairs of vertices and for each pair it takes $O(n)$ time to check if it is pseudo-dominating. Similarly, checking if Rule~2 can be applied takes $O(n^5)$ time.
As each application of each of these rules takes $O(n)$ time, and reduces the size of $G$, this procedure will complete in $O(n^6)$ time. By repeated use of Lemma~\ref{l-rules}, this results in a cover-complete triple  $(G',J',y')$ 
that satisfies the two properties of the lemma; in particular given a a smallest connected vertex cover~$S'$ of~$G'$ that contains~$J'$, it is possible to find in $O(n^{2s+17})$ time a smallest connected vertex cover~$S$ of~$G$ that contains~$J$, as we applied Rules~1 and~2 at most $n$ times and by condition~3 we need $O(n^{2s+16})$ time per application.
Moreover,~$G'$ contains no pseudo-dominating pair with a common neighbour in~$L'=L_{J'}$ and $G'[L']$ is $(P_1+K_4)$-free, as otherwise we could still apply Rule~1 or Rule~2, respectively.
Hence $(G',J',y')$ is a free cover-complete triple. \qed
\end{proof}

Let $(G,J,y)$ be a free  cover-complete triple. A connector of a connected vertex cover~$S$ of $G$ is {\it minimal} if it does not properly contain a smaller connector of~$S$. The next three lemmas are on free cover-complete triples; the second makes use of the first.

\begin{lemma}\label{l-claim1}
Let $(G,J,y)$ be a free  cover-complete triple.  Then every minimal connector~$L^*$ of every type~2 connected vertex cover~$S$ of $G$ is a clique.
\end{lemma}

\begin{proof}
For contradiction, suppose that $L^*$ is not a clique. Then $L^*$ contains two non-adjacent vertices $w_1$ and $w_2$. As $L^*$ is a minimal connector, $w_1$ has a neighbour in $J$ not adjacent to $w_2$, and vice versa. However, then $(w_1,w_2)$ is a pseudo-dominating pair of $G$. This is not possible, as $S$ is of type~2.\qed
\end{proof}

\begin{lemma}\label{l-claim2}
Let $(G,J,y)$ be a free  cover-complete triple that has a pseudo-dominating pair $(w_1,w_2)$. Then every minimal connector~$L^*$ of every type~2 connected vertex cover~$S$ of $G$ has size at most~$5$.
\end{lemma}

\begin{proof}
For contradiction, suppose that $|L^*|\geq 6$. By Lemma~\ref{l-claim1}, $L^*$ is a clique.  As $(G,J,y)$ is free, $G'[L']$ is $(K_4+P_1)$-free by definition. Hence $w_1$ must be adjacent to least three vertices of $L^*$, which we denote by $x_1,x_2,x_3$. Note that $\{w_1,x_1,x_2,x_3\}$ induces a $K_4$ in $G[L]$. By definition of a pseudo-dominating pair, $w_1$ and $w_2$ are non-adjacent. As $(G,J,y)$ is free, $w_2$ is not adjacent to any neighbour of $w_1$ in $L$ by definition. Hence $w_2$ is not adjacent to any  vertex of $\{x_1,x_2,x_3\}$. This means that the set $\{w_1,w_2,x_1,x_2,x_3\}$ induces a $K_4+P_1$ in $G[L]$, a contradiction. \qed
\end{proof}

\begin{lemma}\label{l-algo}
Let $(G,J,y)$ be a free  cover-complete triple that has no pseudo-dominating pair. It is possible to find in $O(n^3)$ time a clique $K\subseteq L$ with $N_{G}(K)\cap J=J$.
\end{lemma}

\begin{proof}
We describe how to construct $K$. Consider a vertex $w_1\in L$ that has maximal neighbourhood in $J$, that is, there is no vertex $w\in L$ with $N_{G}(w_1)\cap J\subsetneq N_{G}(w)\cap J$. We put $w_1$ in $K$. Suppose that at some point we have constructed a clique $K=\{w_1,\ldots,w_i\}$ for some $i\geq 1$. If $N_{G}(K)\cap J=J$, then we stop. Otherwise we pick a vertex $w_{i+1}$ with maximal neighbourhood in $J\setminus N_{G}(K)$ over all vertices in $L$ (or equivalently, all vertices in $L\setminus \{w_1,\ldots,w_i\}$). Note that $w_{i+1}$ exists as $G$ is connected.

Suppose that $w_{i+1}$ is adjacent to some $x\in N_{G}(K)\cap J$. Then, by Lemma~\ref{l-ind}, we find that $x$ is adjacent to a unique vertex $w_h$ in $K$. By the same lemma, $w_{i+1}$ is not adjacent to $w_h$. As $G$ has no pseudo-dominating pair and $w_{i+1}$ has a neighbour in $J\setminus N_{G}(K)$ (that is, a neighbour not adjacent to $w_h$), we find that $N_{G}(w_h)\subsetneq N_{G}(w_{i+1})$. This means that we would have chosen $w_{i+1}$ earlier, namely instead of $w_h$.  Hence, $w_{i+1}$ is not adjacent to any $x\in N_{G}(K)\cap J$. As $G$ has no pseudo-dominating pairs, this means that $w_{i+1}$ is adjacent to every $w_j$ with $1\leq j\leq i$. That is, we can extend $K$ into a larger clique by adding $w_{i+1}$.

As we increase $N_{G}(K)\cap J$ each time we add a new vertex to $K$, our procedure will stop with the desired output $K=\{w_1,\ldots,w_r\}$ for some $r\geq 1$. We note that constructing $K$ takes $O(n^3)$ time.\qed
\end{proof}

We are now ready to prove the following theorem. 

\begin{theorem}\label{t-ind2}
For every $s\geq 0$, {\sc Connected Vertex Cover Completion} can be solved in $O(n^{2s+19})$ time for 
cover-complete triples $(G,J,y)$, where $G$ is an $(sP_1+P_5)$-free graph.
\end{theorem}

\begin{proof}
Let $s\geq 0$ and let $(G,J,y)$ be a  cover-complete triple, where $G$ is an $(sP_1+P_5)$-free graph. We first apply Lemma~\ref{l-eventual} to obtain a free  cover-complete triple $(G',J',y')$ in $O(n^6)$ time. By the same lemma, $G'$ is $(sP_1+P_5)$-free. Our aim is to find a smallest connected vertex cover of $G'$ that contains $J'$ in 
polynomial time, so that we can apply statement~2 of Lemma~\ref{l-eventual}. We first compute in $O(n^{2s+16})$ time a smallest type~1 connected vertex cover $S^*$ of $G'$ using Lemma~\ref{l-first}. We now need to compute a smallest type~2 connected vertex cover~$S'$ of $G'$ and compare $|S'|$ with $|S^*|$.

We check if $G'$ contains a pseudo-dominating pair. This takes $O(n^3)$ time, as $G'$ contains $O(n^2)$ pairs of vertices and for each pair it takes $O(n)$ time to check if it is pseudo-dominating.

First suppose that $G'$ contains a pseudo-dominating pair. For each set of at most five vertices, we check if it is a minimal connector of size at most~5, and if so we apply Lemma~\ref{l-crucial} on its vertices. This takes $O(n^2)$ time per set.
If we obtain an instance of the form $(G'',\{y''\},y'')$, then we apply Lemma~\ref{l-vcvc}, which takes $O(n^{s+14})$ time. Then we uncontract all contracted edges in $O(n)$ time to get a connected vertex cover of $G'$ of type~2. By Lemma~\ref{l-claim2}, doing this for every possible minimal connector of size at most~5 gives us a smallest type~2 connected vertex cover $S'$ of $G'$. As we process each set of at most five vertices in $O(n^{s+14})$ time and the number of such sets is $O(n^5)$, we find~$S'$ in $O(n^{s+19})$ time. We compare $S'$ and $S^*$ and choose the smaller of the two.

Now suppose that $G'$ has no pseudo-dominating pair.  Let $L'=N_{G'}(J'\setminus \{y'\})$. By Lemma~\ref{l-algo}, we can obtain in $O(n^3)$ time a clique $K\subseteq L'$ with $N_{G'}(K)\cap J'=J'$. Let $K=\{w_1,\ldots,w_r\}$ for some $r\geq 1$. As $K$ is a clique, every vertex cover contains at least $r-1$ vertices of $K$.  We will do as follows: first we will find in $O(n^{s+14})$ time a smallest connected vertex cover of $G'$ that contains $J'\cup K$, and then we will find in $O(n^{s+17})$  time, for $i=1,\ldots,r$, a smallest connected vertex cover of $G'$ that contains $J'\cup (K\setminus \{w_i\})$ and that does not contain~$w_i$. As there are $O(n)$ cases, the total time of processing this case is $O(n^{s+18})$. 

We start by computing a smallest connected vertex cover of $G'$ that contains $J'\cup K$ by set-contracting via each vertex of $K$. This takes $O(n^2)$ time. By Lemma~\ref{l-crucial}, this yields a cover-complete triple $(G'',\{y''\},y'')$ to which we apply Lemma~\ref{l-vcvc} in $O(n^{s+14})$ time. Uncontracting all contracted edges yields, by Lemma~\ref{l-crucial}, a smallest connected vertex cover $S_K$ of $G'$ that contains $J'\cup K$; this takes $O(n)$ time. Hence, the total running time for this step is $O(n^{s+14})$, as we claimed above.

We  now show how to compute, in $O(n^{s+17})$ time, a smallest connected vertex cover of $G'$ that contains $J'\cup 
(K\setminus \{w_1\})$ and that does not contain $w_1$. The cases where $i\geq 2$ are done in the same way.  

Let $A=L'\setminus N_{G'}(w_1)$ consist of all non-neighbours of $w_1$ in $L'$. As $G'[L']$ is $(K_4+P_1)$-free by definition, we find that $G'[A]$ is $K_4$-free. As $w_1$ is not in the connected vertex cover we are looking for we remove $w_1$, and we set-contract, in $O(n^2)$ time, via each neighbour of $w_1$ in $L$. By Lemma~\ref{l-crucial}, we may now consider the resulting  cover-complete triple $(G'',J'',y'')$ where $G''$ is connected and $(sP_1+P_5)$-free.  As $G'$ had no pseudo-dominating pairs, we have that $G''$ has no pseudo-dominating pairs. We write $L''=N_{G''}(J''\setminus \{y''\})$.  As $L''\subseteq A$, we find that $G''[L'']$ is $K_4$-free.

\medskip
\noindent
{\it Claim. Every minimal connector $L^*$ of every connected vertex cover of $G''$ that contains $J''$ has size at most~$3$.}

\medskip
\noindent
We prove the claim by showing that $L^*$ is a clique, which implies that~$L^*$ has size at most~3, as $G''[L'']$ is $K_4$-free. Suppose instead that $L^*$ is not a clique. Then $L^*$ contains two non-adjacent vertices $w_1$ and $w_2$. As $L^*$ is a minimal connector, $w_1$ has a neighbour in $J''$ not adjacent to $w_2$, and vice versa. But then $(w_1,w_2)$ is a pseudo-dominating pair of $G''$: this is not possible, as $G''$ has no pseudo-dominating pairs. This contradiction proves the claim.

\medskip
\noindent
We now consider all subsets in $L''$ that have size at most~3. For each set we check if it is a minimal connector, and if so we apply Lemma~\ref{l-crucial} on its vertices. This takes $O(n^2)$ time per subset. If we obtain an instance $(G''',\{y'''\},y''')$, then we apply Lemma~\ref{l-vcvc} in $O(n^{s+14})$ time. Then uncontracting all contracted edges yields a connected vertex cover of $G''$ that contains $J''$. As there are $O(n^3)$ subsets in $L''$ of size at most~3, the total running time is $O(n^{s+17})$, as we claimed above.
We keep track (in constant time) of the smallest one of these connected vertex covers of $G''$. For this connected vertex cover of $G''$, we uncontract all contracted edges again to obtain a smallest connected vertex cover $S_{w_1}$ of $G'$ that contains $J'\cup (K\setminus \{w_1\})$ and that does not contain $w_1$.

As mentioned, we pick the smallest one out of the connected vertex covers $S_K$ and $S_{w_i}$, $1\leq i\leq r$, to obtain a smallest type~2 connected vertex cover of~$G'$, the size of which we compare with the size of $S^*$. We pick the smallest one.

Thus we obtain in $O(n^6)+O(n^{2s+16})+O(n^3)+O(n^{s+19})+O(n^{s+18})=O(n^{2s+19})$
time a smallest connected vertex cover of $G'$ that contains $J'$ (both in the case  where $G'$ has a pseudo-dominating pair and in the case where $G'$ has no pseudo-dominating pair). As stated, it remains to apply statement~2 of Lemma~\ref{l-eventual}  to find in 
$O(n^{2s+17})$ time a smallest connected vertex cover of $G$ that contains $J$. Hence the total running time is $O(n^{2s+19})$.
The correctness of our algorithm follows immediately from the above case analysis and the description of the cases. \qed
\end{proof}

\section{Our Main Result}\label{s-main2}

In this section we prove Theorem~\ref{t-main}, that is, we show that {\sc Connected Vertex Cover} can be solved in polynomial time for $(sP_1+P_5)$-free graphs for every integer~$s\geq 0$. The proof relies heavily on Theorem~\ref{t-ind2}. The main idea is to reduce an $(sP_1+P_5)$-free input graph~$G$ of {\sc Connected Vertex Cover} to a polynomial number of instances $(G_i,J_i,y_i)$ of {\sc Connected Vertex Cover Completion}. We can then solve each of these instances  $(G_i,J_i,y_i)$ in polynomial time by Theorem~\ref{t-ind2}. Then we translate the resulting connected vertex covers of $G_i$ (which contain $J_i$) into connected vertex covers of $G$. We pick the smallest of these sets as our final output.

We need two more lemmas. We use Lemma~\ref{l-bt} to prove the first one.

\begin{lemma}\label{l-2}
Let $s\geq 0$ and let $G$ be a connected $(sP_1+P_5)$-free graph. Then $G$ has a connected dominating set $D$ that is either a clique or has size at most 
$2s^2+s+3$. Moreover, $D$ can be found in $O(n^{2s^2+s+3})$ time.
\end{lemma}

\begin{proof}
If $G$ is $P_5$-free, then we apply Lemma~\ref{l-bt} to find, in $O(n^3)$ time, a set $D$ that either induces a $P_3$ or is a clique.
Otherwise, as $G$ is $(sP_1+P_5)$-free, there exists an integer $0\leq r\leq s-1$ such that $G$ contains an induced subgraph~$H$ isomorphic to $rP_1+P_5$. Let $V_H=\{a_1,\ldots,a_r,b_1,\ldots,b_5\}$ such that the $b$-vertices induce a $P_5$ in that order. We choose $r$ to be maximum so $G$ contains  no induced $(r+1)P_1+P_5$. Hence, $V_H$ dominates $G$.  As $G$ is $(sP_1+P_5)$-free, $G$ is $P_{5+2s}$-free. Hence, for each $a_i$, there exists a path of at most $5+2s-1$ vertices that connects $a_i$ to $b_1$. Let $H^*$ be the graph that contains $H$ and all these $a_i-b_1$-paths. Then we choose $D=V_{H^*}$. As $V_H$ dominates $G$, we find that $D\supseteq V_H$ also dominates $G$. Moreover, $D$ has size at most $r(5+2s-2)+5\leq 2s^2+s+2$. We can find $D$ by considering, if needed, every set of at most $2s^2+s+2$ vertices in $G$ and by checking if each such a set is dominating. The latter takes $O(n)$ time per set. 
Hence, this brute force procedure takes $O(n^{2s^2+s+3})$ time in total.\qed
\end{proof}

\begin{lemma}\label{l-double}
Let $J$ be an independent set in a connected graph~$G$ such that $J$ has a vertex $y$ that is adjacent to every vertex of $G-J$. Let $J'$ consist of those vertices of $J\setminus \{y\}$ that have two adjacent neighbours in $G-J$ (or equivalently, in $G$).  Then a subset $S$ is a connected vertex cover of $G$ that contains $J$ if and only if $S\setminus J'$ is a connected vertex cover of $G-J'$ that contains $J\setminus J'$.
\end{lemma}

\begin{proof}
Let $w\in J\setminus \{y\}$ be a vertex in~$G$ with two adjacent neighbours $a$ and $b$ in $G-J$ (or equivalently in $G$). Let $S$ be a subset of $G$. First suppose that $S$ is a connected vertex cover of $G$ that contains $J$. Then $S\setminus \{w\}$ is a vertex cover of $G-w$ that contains $J\setminus \{w\}$. As $y\in J$ and $y\neq w$, we find that $S\setminus \{w\}$ contains $y$. Then every vertex of $S\setminus \{w\}$ that belongs to $G-J$ is adjacent to $y$ in $G[S\setminus \{w\}]$.  Moreover, as $S$ is connected and $J$ is independent, every vertex of $J\setminus \{w\}$ must be adjacent in $G[S\setminus \{y\}]$ to a vertex of $G-J$. Hence, $S\setminus \{w\}$ is connected in $G-w$.

Now suppose that $S\setminus \{w\}$ is a connected vertex cover of $G-w$ that contains $J\setminus \{w\}$. Then $S$ is a vertex cover of $G$ that contains $J$. As $y\in J$, we find that $S$ contains $y$. As $ab$ is an edge, $S$ contains at least one of $a$ and $b$. Then $w$ and $y$ are adjacent in $S$ either due to the edges $ya$, $aw$ (if $a$ is in $S$) or due to the edges $yb$, $bw$ (if $a$ is not in $S$, as then $b\in S$). Hence $S$ is connected in $G$.

We now consider the graph $G-w$ and repeat the arguments above for any vertex in $J'\setminus \{w\}$. \qed
\end{proof}

We are now ready to prove our main result. 

\medskip
\noindent
{\bf Theorem~\ref{t-main}. (Restated)}
{\it For every $s\geq 0$, {\sc Connected Vertex Cover} can be solved in $O(n^{21s^3 + 26})$ time for $(sP_1+P_5$)-free graphs.}

\begin{proof}
Let $G$ be an $(sP_1+P_5)$-free graph on $n$ vertices for some $s\geq 0$. We may assume without loss of generality that $G$ is connected. By Lemma~\ref{l-2} we can first compute in $O(n^{2s^2+s+3})$ time a connected dominating set~$D$ that either has size at most $2s^2+s+3$ or is a clique.  We note that, if $D$ is a clique, any vertex cover of~$G$ contains all but
at most 
one vertex 
of $D$. This leads to a case analysis where we guess the subset $D^*\subseteq D$ of vertices not in a minimum connected vertex cover of $G$. 
That is, we choose a set of at most one vertex if $D$ is a clique and a set of at most $|D|$ vertices otherwise, and eventually look at all such sets.
As $|D|\leq 2s^2+s+3$ if $D$ is not a clique, the number of guesses is $O(n^{2s^2+s+3})$. For each guess of $D^*$, we compute a smallest connected vertex cover~$S_{D^*}$ that contains all vertices of $D\setminus D^*$ and no vertex of $D^*$. Then, in the end, we return one that has minimum size overall. 

Let $D^*$ be a guess. Before we start our case analysis we first prove the following claim.

\medskip
\noindent
{\it Claim~1. We may assume, at the expense of an $O(n^{16s^3+4})$ factor in the running time, that $D\setminus D^*$ is connected.}

\medskip
\noindent
We prove Claim~1 as follows. Suppose $D\setminus D^*$ is not connected. Recall that $G[D]$ is either a complete graph or has size at most $2s^2+s+3$. In the first case, $G[D\setminus D^*]$ is connected. Hence, the second case applies so $D$ has size at most $2s^2+s+3$. Let $v\in D\setminus D^*$. As $G$ is $(sP_1+P_5)$-free, $G$ is also $P_{5+2s}$-free. Hence, for each $u\in D\setminus (D^*\cup \{v\})$, any connected vertex cover of $G$ contains a path of at most $5+2s-1$ vertices that connects $u$ to $v$. We will guess all these $u-v$-paths (using only vertices from $G-D^*$) and add their vertices to $D$. As the number of paths is at most $2s^2+s+2$,  this branching adds an 
$O(n^{(5+2s-3)(2s^2+s+2)})=O(n^{16s^3+4})$ factor to our running time and increases our set~$D$ by at most 
$24s^3$ extra vertices. We have proven Claim~1.

\medskip
\noindent
{\bf Case 1.} $D^*=\emptyset$.\\
We compute a minimum vertex cover $S'$ of $G-D$ in polynomial time by Theorem~\ref{t-vc}. 
To be more precise, this takes  $O(n^{s+14})$ time by using the same arguments as in the proof of Lemma~\ref{l-vcvc}. 
Clearly $S'\cup D$ is a vertex cover of $G$. As $D$ is a connected dominating set, $S'\cup D$ is even a connected vertex cover of $G$. Let $S_\emptyset=S'\cup D$. As $S'$ is a minimum vertex cover of $G-D$, $S_\emptyset$ is a smallest connected vertex cover of $G$  that contains all vertices of~$D$. We remember $S_\emptyset$. Note that $S_\emptyset$ is found in  $O(n^{s+14})$ time.

\medskip
\noindent
{\bf Case 2.} $1\leq |D^*|\leq |D|$\; (recall that $|D|\leq 2s^2+s+3$).\\
Recall that we are looking for a smallest connected vertex cover of $G$ that contains every vertex of $D\setminus D^*$ but does not contain any vertex of $D^*$. Hence~$D^*$ must be an independent set and $G-D^*$ must be connected (if one of these conditions is false, then we stop considering the guess $D^*$). Moreover, a vertex cover that contains no vertex of $D^*$ must contain all vertices of $N_G(D^*)$. Hence we can safely contract not only any edge between two vertices of $D\setminus D^*$, but also any edge between two vertices in $N_G(D^*)$ or between a vertex of $D\setminus D^*$ and a vertex in $N_G(D^*)$. We perform edge contractions recursively and as long as possible while remembering all the edges that we contract. 
This takes $O(n)$ time.
Let $G^*$ be the resulting graph. 

Note that the set $D^*$ still exists in $G^*$, as we did not contract any edges with an endpoint in $D^*$. By Claim~1, the set $D\setminus D^*$ in $G$ corresponds to exactly one vertex of~$G^*$. We denote this vertex by~$y$. We observe the following equivalence, which is obtained after uncontracting all the contracted edges.

\medskip
\noindent
{\it Claim~2. Every smallest connected vertex cover of $G^*$ that contains $y$ and that does not contain any vertex of $D^*$ corresponds to a smallest connected vertex cover of $G$ that contains $D\setminus D^*$ and that does not contain any vertex of $D^*$, and vice versa.}

\medskip
\noindent
As we obtained $G^*$ in $O(n)$ time, and we can uncontract all contracted edges in $O(n)$ time as well, Claim~2 tells us that we may consider $G^*$ instead of $G$. As $G$ is connected and $(sP_1+P_5)$-free, $G^*$ is connected and $(sP_1+P_5)$-free as well by Lemma~\ref{l-contract}.

We write $J^*=N_{G^*}(D^*)$ and note that $y$ belongs to $J^*$ as $D$ is connected in~$G$. We now consider the graph $G^*-D^*$.  As $G-D^*$ is connected,  $G^*-D^*$ is connected. By Claim~2,  our new goal is to find a smallest connected vertex cover  of $G^*-D^*$ that contains $J^*$. By our procedure, $J^*$ is an independent set of $G^*-D^*$. As $D$ dominates $G$, we find that $D\setminus D^*$ dominates every vertex of $G-D^*$ that is not adjacent to a vertex of $D^*$. Hence the vertex $y$, which corresponds to the set $D\setminus D^*$, is adjacent to every vertex of $(G^*-D^*)-J^*$ in the graph $G^*-D^*$. 

Let $J\subseteq J^*$ consist of $y$ and those vertices in $J^*$ whose neighbourhood in $G^*-D^*$ is an independent set. As $y$ is adjacent to every vertex of $(G^*-D^*)-J^*$  in $G^*-D^*$, and we can remember the set $J^*\setminus J$, we can apply Lemma~\ref{l-double} and remove $J^*\setminus J$. That is, it suffices to find a smallest connected vertex cover of the graph $G'=(G^*-D^*)-(J^*\setminus J)$ that contains $J$.

As $J^*$ is an independent set of $G^*-D^*$, we find that $J$ is an independent set of $G'$. By definition, $y\in J$. As $y$ is adjacent to every vertex of $(G^*-D^*)-J^*$  in $G^*-D^*$, we find that $y$ is adjacent to every vertex in $G'-J$. By definition, the neighbours of each vertex in $J \setminus \{y\}$  form an independent set in $G'-J$. Hence the triple $(G',J,y)$ is  cover-complete. This means that we can apply Theorem~\ref{t-ind2} to find in $O(n^{2s+19})$ time a smallest connected vertex cover $S'$ of $G'$ that contains~$J$.

We translate $S'$ in constant time  into a smallest connected vertex cover $S^*$ of $G^*-D^*$ that contains $J^*$ by adding $J^*\setminus J$ to $S'$. We translate $S^*$ in $O(n)$ time into a smallest connected vertex cover $S_{D^*}$ of $G$ that contains no vertex of $D^*$ by uncontracting any contracted edges. It takes $O(n^{2s+19})$ time to find he time $S_{D^*}$.

\medskip
\noindent
As mentioned, in the end we pick a smallest set of the sets $S_{D^*}$. This set is then a minimum connected vertex cover of $G$.
As there are  $O(n^{2s^2+s+3}\cdot n^{16s^3+4})$ of such sets, each of which is found in $O(n^{2s+19})$ time, the total running time is $O(n^{21s^3 + 26})$. 
The correctness of our algorithm follows immediately from the above case analysis and the description of the cases. \qed
\end{proof}

Note that the algorithm in Theorem~\ref{t-main} not only solves the decision problem, but also finds a minimum connected vertex cover of a given $(sP_1+P_5)$-free graph. 

\section{Weighted Connected Vertex Cover}\label{s-weight}

Let $G=(V,E)$ be a {\it vertex-weighted} graph, that is, each vertex~$v$ of $G$ has an associated non-negative weight $w_v$. 
The {\it weight} of a subset~$S\subset V$ is defined as $w(S)=\sum_{v\in S}w_v$.
A vertex cover~$S$ of $G$ is a {\it minimum weight vertex cover}  if $G$ has no vertex cover~$S'$ with $w(S')<w(S)$.   The {\sc Weighted Vertex Cover} problem is to find a minimum weight vertex cover of a vertex-weighed graph~$G$.  
As mentioned, Theorem~\ref{t-vc} can be generalized to hold for {\sc Weighted Vertex Cover}~\cite{GKPP17}. As we use Theorem~\ref{t-vc} to prove Theorem~\ref{t-main}, this allows us to solve the following more general problem in polynomial time for $(sP_1+P_5)$-free graphs ($s\geq 0$); 
note that we formulate this generalization as an optimization problem.

\optproblemdef{{\sc Weighed Connected Vertex Cover}}{a graph $G$, an integer $k$ and a non-negative vertex weight function~$w$.}{find a minimum weight connected vertex cover of $G$.} 

In order to prove this result we first need to generalize the {\sc Connected Vertex Cover Completion} problem. 

\optproblemdef{{\sc Weighted Connected Vertex Cover Completion}}{a cover-complete triple $(G,J,y)$, where $G$ has a non-negative vertex weight function~$w$.}{find a minimum weight connected vertex cover $S$ of $G$ that contains~$J$.} 

\noindent
We first prove the following theorem.

\begin{theorem}\label{t-ind2weight}
For every $s\geq 0$, {\sc Weighted Connected Vertex Cover Completion} can be solved in polynomial time for
cover-complete triples $(G,J,y)$, where $G$ is an $(sP_1+P_5)$-free graph with a non-negative vertex weight function~$w$.
\end{theorem}

\begin{proof}
We can follow the same approach as in the proof of Theorem~\ref{t-ind2}. We first note that Lemma~\ref{l-contract} is a structural lemma unrelated to the vertex weight function~$w$. Lemma~\ref{l-bt} was not needed for the proof of Theorem~\ref{t-ind2} and we do not need it here either. For Lemma~\ref{l-crucial}, we do not have to adjust statements~1 and~2 and only have to replace statement~3 by its weighted version. In order to do so, we define the weight of the new vertex $y_w$, obtained from set-contracting via a vertex~$w$, as the sum of the weights of all the vertices in $J_w\cup \{w\}$. We can then use the same arguments. Lemmas~\ref{l-ind}--\ref{l-pseudotriple} are structural lemmas that are unrelated to the vertex weight function~$w$, so we can still use them. We need to replace Lemma~\ref{l-vcvc} by its weighted version. We can then use the same arguments; in particular, as we may replace
Theorem~\ref{t-vc} by its weighted version~\cite{GKPP17}. We can also replace Lemma~\ref{l-first} by its weighted version: its proof uses brute force searching, and instead of remembering and updating the smallest size of a connected vertex cover, we keep track of the smallest weight. Lemma~\ref{l-rules} still holds in our setting as well. That is, after replacing condition~3 by its weighted version, we can still use the same arguments (modified for weights of sets instead of their sizes). The same holds for Lemma~\ref{l-eventual} (we need to replace property~2). Lemmas~\ref{l-claim1} and~\ref{l-claim2} are structural lemmas unrelated to the vertex weight function~$w$, so we can still use them. Lemma~\ref{l-algo} is algorithmic, but as this lemma is not related to vertex weight functions we can still use it. 
That is, any clique $K\subseteq L$ with $N_G(K)\cap J=J$ found by Lemma~\ref{l-algo} suffices, as every (connected) vertex cover must use all but at most one vertices of a clique. 
Hence, for proving Theorem~\ref{t-ind2weight} we can use the same arguments as in the proof of Theorem~\ref{t-ind2}; in particular the claim inside the proof of Theorem~\ref{t-ind2} is still valid and instead of remembering the smallest size of the vertex covers found by the algortihm so far, we remember the smallest weight.
\qed
\end{proof}

\noindent
We are now ready to show the following result.

\begin{theorem}\label{t-mainweight}
For every $s\geq 0$, {\sc Weighted Connected Vertex Cover} can be solved in polynomial time for $(sP_1+P_5$)-free vertex-weighted graphs.
\end{theorem}

\begin{proof}
Let $s\geq 0$, and let $G$ be an $(sP_1+P_5)$-free graph with a non-negative vertex weight function~$w$.
We first recall that Lemma~\ref{l-contract} is unrelated to the vertex weight function~$w$. The same holds for Lemma~\ref{l-bt}. Hence we may still use both lemmas. In particular this implies that Lemma~\ref{l-2} still holds.
Lemma~\ref{l-double} is a structural lemma that is unrelated to the vertex weight function~$w$, so we can safely use it. By these observations and Theorem~\ref{t-ind2weight}, we can now follow the same arguments as used in
the proof of Theorem~\ref{t-main}. This proof is based on brute force searching. The only thing we need to do is to remember the smallest weight of the vertex covers found during the execution of the algorithm instead of their sizes.
\qed
\end{proof}

\section{Conclusions}\label{s-con}

We proved that {\sc (Weighted) Connected Vertex Cover} is polynomial-time solvable for $(sP_1+P_5)$-free graphs for every integer $s\geq 0$. 
We finish our paper by posing the following two open problems.

\begin{enumerate}
\item  What is the complexity of {\sc Connected Vertex Cover} for $P_6$-free graphs?
\item Does there exist an integer~$r$ such that {\sc Connected Vertex Cover} is \NP-complete for $P_r$-free graphs?
\end{enumerate}
For Question~1, it might be easier to consider first the class of $(P_2+P_3)$-free graphs, for which we do not know the complexity of {\sc Connected Vertex Cover} either.
For Question~2, we need a better understanding of $P_r$-free graphs. The {\sc Connected Vertex Cover} problem belongs to a range of problems which we only know to be polynomial-time solvable on $P_r$-free graphs up to some value of~$r$. 
These problems include {\sc Vertex Cover}, {\sc Feedback Vertex Set}, {\sc Connected Feedback Vertex Set}, {\sc Independent Feedback Vertex Set}, {\sc Odd Cycle Transversal}, {\sc Connected Odd Cycle Transversal}, {\sc Independent Odd Cycle Transversal}, $3$-{\sc Colouring} and {\sc (Dominating) Induced Matching}, see~\cite{BDFJP17,GJPS17} for further details. Even our understanding of bipartite $P_r$-free graphs is limited. For instance, we only know that {\sc Hypergraph 2-Colourability} is polynomial-time solvable on $P_7$-free incidence graphs (which are bipartite)~\cite{CS16}.

\medskip
\noindent
{\it Acknowledgements.} We thank an anonymous reviewer of the conference version of our paper for helpful comments.


\begin{thebibliography}{99}

\bibitem{Al04}
V. E. Alekseev, Polynomial algorithm for finding the largest independent sets in graphs without forks, Discrete Applied Mathematics 135 (2004) 3--16.

\bibitem{BT90}
G.~Bacs\'{o} and Zs.~Tuza,
Dominating cliques in $P_5$-free graphs,
Periodica Mathematica Hungarica 21(1990) 303--308.

\bibitem{BNP96}
V. Balachandhran, P. Nagavamsi, C. Pandu Rangan, Clique transversal and clique independence on
comparability graphs, Information Processing Letters 58 (1996) 181--184.

\bibitem{BY89}
E. Balas and C. S. Yu, On graphs with polynomially solvable maximum-weight clique problem, Networks 19 (1989) 247--253.

\bibitem{BDFJP17}
M. Bonamy, K.K. Dabrowski, C. Feghali, M. Johnson and D. Paulusma, Independent feedback vertex set for $P_5$-free graphs, Proc. ISAAC 2017, Leibniz International Proceedings in Informatics, to appear.

\bibitem{BM18}
A. Brandst\"adt and R. Mosca,
Maximum weight independent set for $\ell$claw-free graphs in polynomial time, Discrete Applied Mathematics 237 (2018) 57--64.

\bibitem{CCFS14}
E. Camby, J. Cardinal, S. Fiorini and O. Schaudt,
The price of connectivity for vertex cover, Discrete Mathematics \& Theoretical Computer Science 16 (2014) 207--224.

\bibitem{CS16}
E. Camby and O. Schaudt, A new characterization of $P_k$-free graphs, Algorithmica 75 (2016), 205--217.

\bibitem{CL10}
J. Cardinal and E. Levy,
Connected vertex covers in dense graphs,
Theoretical Computer Science 411 (2010) 2581--2590.
 
\bibitem{CHJMP18}
N. Chiarelli, T.R. Hartinger, M. Johnson, M. Milanic and D. Paulusma, Minimum connected transversals in graphs: new hardness results and tractable cases using the price of connectivity, 
Theoretical Computer Science 705 (2018) 75--83.

\bibitem{CGKP15}
J.-F. Couturier, P.A. Golovach, D. Kratsch and D. Paulusma, List Coloring in the Absence of a Linear Forest, 
Algorithmica 71 (2015) 21--35.

\bibitem{EGM10}
B. Escoffier, L. Gourv\`es and J. Monnot,
Complexity and approximation results for the connected vertex cover problem in graphs and hypergraphs,
Theoretical Computer Science 8 (2010) 36--49.

\bibitem{FM09}
H. Fernau and D. Manlove, Vertex and edge covers with clustering properties: complexity and algorithms, Journal of Discrete Algorithms 7 (2009) 149--167.

\bibitem{GJ77}
M.R. Garey, D.S. Johnson, The rectilinear Steiner tree problem is \NP-complete, SIAM Journal on Applied Mathematics 32 (1977) 826--834.

\bibitem{Ga74}
F. Gavril, The intersection graphs of subtrees in trees are exactly the chordal graphs,
Journal of Combinatorial Theory, Series B 16 (1974) 47--56.

\bibitem{GH09}
P.~A. Golovach and P.~Heggernes,
Choosability of {$P_5$}-free graphs,
Proc. MFCS 2009, Lecture Notes in Computer Science 5734 (2009) 382--391.

\bibitem{GJPS17}
P.A. Golovach, M. Johnson, D. Paulusma and J. Song, A survey on the computational complexity of colouring graphs with forbidden subgraphs, Journal of Graph Theory 84 (2017) 331--363.

\bibitem{GKPP17}
A.~Grzesik, T.~Klimo\v{s}ov\'a, M.~Pilipczuk, and M.~Pilipczuk,
Polynomial-time algorithm for maximum weight independent set on
  {$P_6$}-free graphs, {\em CoRR}, abs/1707.05491, 2017.
  
\bibitem{GP00}
V. Guruswami and C. Pandu Rangan, Algorithmic aspects of clique-transversal
and clique-independent sets, Discrete Applied Mathematics 100 (2000) 183--202.

\bibitem{HJMP16}
T.R. Hartinger, M. Johnson, M. Milanic and D. Paulusma, The price of connectivity for transversals, European Journal of Combinatorics 58 (2016) 203--224.

\bibitem{HKLSS10}
C.~T. Ho\`ang, M.~Kami\'nski, V.~V. Lozin, J.~Sawada, and X.~Shu,
Deciding $k$-colorability of {$P_5$}-free graphs in polynomial time,
Algorithmica 57 (2010) 74--81.

\bibitem{HPW09}
P.~van 't Hof, D.~Paulusma and G.J. Woeginger,
Partitioning graphs in connected parts,
Theoretical Computer Science 410 (2009) 4834--4843.

\bibitem{JPP18}
M. Johnson, G. Paesani and D. Paulusma, Connected vertex cover for $(sP_1+P_5)$-free graphs, Proc. WG 2018,
Lecture Notes in Computer Science, to appear.

\bibitem{KP}
W. Kern and D. Paulusma, Contracting to a longest path in H-free graphs, Manuscript.

\bibitem{KKTW01}
D.~Kr{\'a}l', J.~Kratochv\'{\i}l, {\relax Zs}.~Tuza, and G.~J. Woeginger, Complexity of coloring graphs without forbidden induced subgraphs, Proc. WG 2001, Lecture Notes in Computer Science 2204 (2001) 254--262.

\bibitem{LYW17}
Y. Li, Z. Yang and W. Wang,
Complexity and algorithms for the connected vertex cover problem in 4-regular graphs,
Applied Mathematics and Computation 301 (2017) 107--114.
 
\bibitem{LVV14}
D.~Lokshtanov, M.~Vatshelle, and Y.~Villanger, Independent set in {$P_5$-free} graphs in polynomial time, Proc. SODA 2014, pages 570--581.

\bibitem{Mi80}
G.J. Minty, On maximal independent sets of vertices in claw-free graphs, Journal of Combinatorial Theory, Series B 28 (1980)
284--304.

\bibitem{Mo12}
R. Mosca,
Stable sets for $(P_6, K_{2, 3})$-free graphs, Discussiones Mathematicae Graph Theory 32 (2012) 387--401.

\bibitem{Mu17}
A. Munaro, Boundary classes for graph problems involving non-local properties, Theoretical Computer Science 692 (2017) 46--71.

\bibitem{Po74}
S. Poljak, A note on stable sets and colorings of graphs,
Commentationes Mathematicae Universitatis Carolinae. 15 (1974) 307--309.

\bibitem{PH08}
P.K. Priyadarsini and T. Hemalatha, Connected vertex cover in 2-connected planar graph with maximum degree 4 is \NP-complete, International Journal of Mathematical, Physical and Engineering Sciences. 2 (2008) 51--54.

\bibitem{TIAS77}
S. Tsukiyama, M. Ide, H. Ariyoshi, and I. Shirakawa,
A new algorithm for generating all the maximal independent sets,
SIAM Journal on Computing 6 (1977) 505--517.

\bibitem{Sh80}
N. Sbihi, Algorithme de recherche d'un stable de cardinalit\'e maximum dans un graphe sans \'etoile,
Discrete Mathematics 29 (1980) 53--76.

\bibitem{UKG88}
S. Ueno, Y. Kajitani and S. Gotoh,
On the nonseparating independent set problem and feedback set problem for graphs with no vertex degree exceeding three, Discrete Mathematics 72 (1988) 355--360.

\bibitem{WKO91}
T. Wanatabe, S. Kajita and K. Onaga, Vertex covers and connected vertex covers in 3-connected graphs, Proc. IEEE International Symposium on Circuits and Systems 1991, 1017--1020.

\end{thebibliography}
\end{document}